\documentclass[conference]{IEEEtran}
\ifCLASSINFOpdf
\else
\fi
\usepackage{amsfonts}
\usepackage[dvips]{graphicx}
\usepackage{times}
\usepackage[]{cite}
\usepackage{amsmath}
\usepackage{array}
\usepackage{amssymb}

\usepackage{stfloats}
\usepackage{slashbox}
\usepackage{graphicx}
\usepackage{footnote}
\usepackage{booktabs}
\usepackage{array}
\usepackage{algorithmic}
\usepackage{algorithm}
\usepackage{subeqnarray}
\usepackage{cases}
\usepackage{threeparttable}
\usepackage{color}
\usepackage[numbers,sort&compress]{natbib}
\IEEEoverridecommandlockouts
\newtheorem{theorem}{Theorem}

\newtheorem{proof}{Proof}

  \usepackage{flafter}

\usepackage{subfig}
\usepackage{float}

\usepackage{caption}
\begin{document}
%
\title{Energy Efficiency in Multicast Multihop D2D Networks}
\IEEEoverridecommandlockouts

\author{\IEEEauthorblockN{Zicheng Xia,
Jiawei Yan, and Yuan Liu}
 \vspace*{0.5em}
\IEEEauthorblockA{School of Electronics and Information Engineering,
South China University of Technology, Guangzhou 510641, China}
Emails: zicheng-xia@foxmail.com, jiawei.yan@yahoo.com, eeyliu@scut.edu.cn
%

}



%


\maketitle


\begin{abstract}
As the demand of mobile  devices (MDs) for data services is explosively increasing, traditional offloading in the cellular networks is facing the contradiction of energy efficiency and quality of service. Device-to-device (D2D) communication is considered as an effective solution. This work investigates a scenario where the MDs have the same demand for common content and they cooperate to deliver it using multicast multihop relaying. We focus on the problem of  total power minimization by grouping the MDs in multihop D2D networks, while maintaining the minimum rate requirement of each MD. As the problem is shown to be NP-complete and the optimal solution can not be found efficiently, two greedy algorithms are proposed to solve this problem in polynomial time. Simulation results demonstrate that lots of power can be saved in the content delivery situation using multihop D2D communication, and the proposed algorithms are suitable for different situations with different advantages.
\end{abstract}


%
\IEEEpeerreviewmaketitle

\section{Introduction}
Recent communication systems are confronted with explosive growth in mobile applications. It is reported that people more frequently use their mobile devices (MDs) to connect, follow social media, watch live shows, etc. Therefore, the  increase of data demands requires efficient cellular technologies to remain users' quality of service. The traditional cellular network technique is not capable of meeting future's service needs.

	Offloading cellular networks is highly attracting attention in recent years by either migrating to new network topologies or developing enhancement techniques of current cellular networks to accommodate more subscribers with higher data rates\nocite{Aijaz2013A,sciancalepore2016offloading,kang2014mobile,Han2012Mobile,Y2016Optimal,pyattaev2014network,Al2014Optimal,sheng2016energy}	\cite{Aijaz2013A}-\cite{sheng2016energy}. A number of cellular offloading techniques have been proposed in the literature, such as switching to femtocells \cite{yun2015distributed}, \cite{Calin2010On} or Wi-Fi networks \cite{kang2014mobile}, \cite{rossi2015cooperative}. These techniques have a range of advantages including low cost, standardized interface and high quality of services. Another technique would be content multicast to MDs on the cellular network \cite{Y2016Optimal}, \cite{cao2015social}. This technique performs well in high density of MDs  requiring the same content. Nevertheless, a major problem of multicast is the limitation that all MDs download content by an identical rate decided by the worst channel among all MDs, which may sacrifice the performance of MDs in better channel conditions. Another attractive technique is to deliver content cooperatively in which MDs act as relays and connect other MDs with no congestion \cite{Aijaz2013A}, \cite{pyattaev2014network},\nocite{liu2014interference,Whitbeck2012Push,Wang2012Content,Ristanovic2011Energy,Al2011Offloading,Li2011Multiple}
 \cite{liu2014interference}-\cite{Li2011Multiple}.

	The concept of device-to-device (D2D) communications is to establish direct links between MDs and bypass the base station (BS). Cooperative D2D communication shows a good potential for content delivery (e.g., files, videos, live shows, etc.) \cite{kang2014mobile},\nocite{liu2014interference,Whitbeck2012Push,Wang2012Content,Al2011Offloading}
 \cite{liu2014interference}-\cite{Al2011Offloading}. The D2D-enabled cellular offloading can be realized by the following ways. One is that BS divides the content into different chunks and distributes them to different MDs. Then MDs exchange the chunks via an ad hoc manner until all MDs receive the complete content. Consequently, only a few copies of the content are delivered by cellular network instead of the entire content  \cite{Al2014Optimal}. Another direction is group-based multicast in multihop D2D networks where content is relayed group by group using multicast. However, this scenario is much less investigated in the literature.

	The idea of multicast in multihop networks has many advantages. First, it can considerably offload traffic of BS. Most importantly, the total power consumption can be significantly reduced  since the distances between MDs are much shorter than the distances between BS and MDs. Moreover, edge MDs far from BS can gain much better transmission from their close multicast groups. However, this idea depends on how to group MDs and form the multicast links. This is very challenging as MD grouping usually falls into NP-complete problem even in a single-hop case.

	Several  works have been done on multicast D2D networks. For example, the authors in \cite{Militano2014Wi} compared Wi-Fi cooperation and D2D-based multicast content distribution in terms of time-saving and power-saving. However, how to group MDs was not discussed. In \cite{Seppala2011Network}, grouping multicast was considered for delay and throughput  problems. In the work, the groups were assumed to be fixed. In \cite{Al2014Optimal}, fairness constraint was imposed to user grouping for channel allocation.

	To our knowledge, there is no work on total power minimization with group multicast in cooperative multihop D2D networks. Our aim is to find an efficient power minimization strategy to group MDs by multicast transmission while maintaining rate requirements of all MDs. The multicast groups are connected via multihop relaying fashion. The problem has no optimal solution with polynomial time complexity. Thus we alternatively propose two heuristic algorithms to balance the performance and complexity. Each algorithm has its own
\begin{figure}[t]
\centering
\includegraphics[scale=.3]{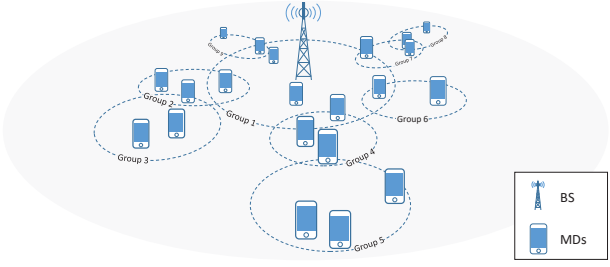}
 \caption{Content delivery in multicast multihop D2D network.}\label{fig:ao1}
\end{figure}	
\noindent advantages and disadvantages.  The simulation results verify the effectiveness of the proposed algorithms.

	The rest of the paper is organized as follows.  Section II describes the system model, including main parameters and problem formulation. Section III proposes two efficient algorithms. Simulation results and analysis are shown in Section IV. Finally, conclusions are made in Section V.

\section{System Model and Problem Formulation}

\begin{table}[!tp]
\caption{MAIN PARAMETERS AND VARIABLES}
\begin{tabular}{|p{1.5cm}|p{6.4cm}|}
\hline
\multicolumn{2}{|c|}{\textbf{Parameters}}  	\\\hline
$\mathcal {C}$	&The set of MDs $\mathcal{C}=\{1,2,...,C\}$ \\\hline

$P_{t,mn} $ & Transmit power of transmitter $m$ to $n$ \\\hline
$P_{r,mn} $ & Receive power at MD $n$ from transmitter $m$\\\hline
$\beta$ & The path loss exponent \\\hline
$N_0$	&The noise power density \\\hline
$\mathcal{K}_{s_j}$ &The multicast group consists of MDs receiving content from transmitter $s_j$ \\\hline
$R_g(\mathcal{K}_{s_j})$	&The multicast rate of group $\mathcal{K}_{s_j}$ \\\hline
$d_0$&A reference distance of the antenna far-field\\\hline

$h_{mn}$   &The channel quality exponent between MD $m$ and MD $n$\\\hline
$x_{mn}$	&A binary variable that indicates whether the condition of channel between MD $m$  and MD $n$ is the worst in the multicast group transmited by MD $m$  \\\hline

$y_{mn}^h$	&A binary variable indicates whether MD $m$ transmits content to MD $n$ on hop $h$ \\\hline
$R_{min}$	&The required rate of all MDs\\\hline
$H_{max}$	&The maximum of the number of transmit hops  \\\hline
$\mathcal{R}$	&The set of receivers have not yet get content  \\\hline
$\mathcal{S}$	&The set of potential transmitters containing content  \\\hline
$\mathcal{K}$	&The set of disjoint multicast among multiple hops \\\hline

\end{tabular}
\end {table}

We consider a cellular network where the MDs tend to obtain the same content from the single BS, denoted as set $\mathcal{C}=\{1,...,C\}$. As discussed previously, no matter how many services MDs request and how good the wireless environment they are in, BS should always be ready for transmission. Therefore, the load of BS is considerably heavy, especially when MDs are active.

In addition, a MD will consume many resources if it requests for a large content. This kind of data  demands is popular such as file downloading and sharing. Consequently, maintaining sufficient energy storage becomes a great challenge to any cellular system.
To deal with this problem, we consider the scenario that MDs offload the traffic of BS using D2D communication. That is, each MD can receive content either from BS or other MDs through multiple hops as shown in Fig. 1.

The D2D-enabled content distribution process is carried out as follows: Firstly, when new content is available in BS and can be downloaded, the BS chooses a subset of MDs and multicasts the content to them. The required service rates of the selected MDs in the subset should be satisfied. Then, the rest of MDs can obtain the same content from those MDs receiving content from the BS. These MDs form disjoint subsets according to certain grouping rules and each subset chooses its appropriate multicast transmitter. This procedure is repeated until all MDs are grouped and their required rates are maintained. In another word, the content delivery process is carried out group by group using multihop relaying.

\subsection{Parameters and Variables}
We present the main parameters and variables in Table I.

In this paper, we consider the wireless channel gain from MD $m$ to MD $n$ as the following model:
     \begin{align}
     \begin{split}
  h_{mn}(\rm dB)&=\frac{P_{r,mn}(\rm dB)}{ P_{t,mn}{(\rm dB)}} \\&=\underbrace{10\log_{10}K-10\beta\log_{10}\frac{d}{d_0}}_{\rm path loss}-\underbrace{\varphi _{\rm dB}}_{\rm shadow fading}\label{c1}
  \end{split}
    \end{align}
where $P_{t, mn}$ and $P_{r, mn}$ are the transmit power and received power between transmitter $m$ and receiver $n$, respectively; $K$ is a constant which depends on the characteristics of antenna and the attenuation of average channels; $d_0$ is a reference distance (1-10 meters indoors and 10-100 meters outdoors) of the antenna far-field; $\beta$ is the path loss exponent; $\varphi_{\rm dB}$ is Gauss-distributed random variable with mean zero and variance $\sigma_{\varphi_{\rm dB}}^2$; $d$ is the distance between a transmitter and a receiver.



Given the transmitter $m$, the achievable rate of MD $n$ is given by
     \begin{align}
	R_n = \log_2 \left(1+\frac{h_{mn} P_{t,mn}}{N_0}\right),  \label{ac}
    \end{align}
where $N_0$ is the noise power density. For the multicast case, the bit rate $R_n$ should follow the Short Slab theory, which means that the rate is limited by the worst channel. Assume that there is a multicast group $\mathcal{K}{_{s_j}}$ where $s_j$ acts as the transmit node, it uses multicast to serve all MDs. The maximum muticast rate $R_g(\mathcal{K}{_{s_j}})$ is given by

      \begin{align}
      \begin{split}
R_g(\mathcal{K}{_{s_j}}) & =\log_2\left(1+\frac{\sum\limits_{\forall m,n\in\mathcal{K}{_{s_j}}}  h_{mn} P_{t,mn}x_{mn}}{N_0}\right)\\
&= \log_2\left(1+\frac{h_{mw'}P_{t,mw'}}{N_0}\right), \label{ab}
	\end{split}
      \end{align}
where $w'$ is MD with the worst channel quality in multicast group $\mathcal{K}{_{s_j}}$, which means  $R_g(\mathcal{K}{_{s_j}})$ = $\min R_n$, $n\in\mathcal{K}{_{s_j}}$. The required transmit power depends on the decodable rate. That is, for given quality of service demand $R_{min}$. The transmit power is computed by $P_{t,mn}=(2^{R_{min}}-1)N_0/h_{mn}$.



\subsection{ Optimization Problem Formulation }
Our goal is to minimize the total transmit power consumption of the whole network by optimizing multicast group division while maintaining the minimum achievable rate of each MD. The problem can be mathematically formulated as
\renewcommand{\minalignsep}{-80pt}
\begin{flalign}
&\min_{P_t,x,y}  \sum_{m=1}^C \sum_{n=1}^C P_{t,mn} x_{mn} +P_{(BS)} \label{eqn:p1}\\
{\rm s.t.}~~~
& x_{mn} \leq y_{mn}^h, &&\forall m,\forall n\in \mathcal{C} ,2 \leq h \label{eqn:p2}\\
&y_{mn}^h \leq y_{(BS)m}^1  ,&&\forall m, \forall n \in \mathcal{C}, 2 \leq h \label{eqn:p3}\\
&y_{ji}^h \leq \sum_{k=1}^C  y_{kj}^{h'}  ,&&h'=1,..,h-1 \label{eqn:p4}\\
&R_{min}  \leq R_g(\mathcal{K}{_{s_j}}),   &&\forall \mathcal{K}{_{s_j}} \subseteq \mathcal{K}\label{eqn:p5}\\
&\sum_{h=2}^{H_{max}} \sum_{m\neq n}^C y_{mn}^h + y_{(BS)n}^1 =1,&&\forall n \label{eqn:p6}\\
&  h \leq H_{max},		&&\forall h \label{eqn:p7}
\end{flalign}
where $x_{mn}$ is a binary variable that indicates whether the channel condition between MD $m$  and MD $n$ is the worst in the multicast group transmitted by MD $m$; $y_{mn}^h$ is a binary variable indicates whether MD $m$ transmits content to MD $n$ on hop $h$; $R_{min}$ is the minimum rate requirement of all MDs to decode the same content; $R_g(\mathcal{K}{_{s_j}})$ is the multicast rate in the multicast group $\mathcal{K}{_{s_j}}$; $H_{max}$ is the predefined maximum tolerated hops to some extent delay considerations. The subscript BS represents the base station. 

The objective \eqref{eqn:p1} minimizes total transmit power of MDs and BS.
Constraint \eqref{eqn:p2} ensures that the transmission rate determined by the worst channel condition in a multicast group. Constraint \eqref{eqn:p3} ensures that MD $n$ can transmit on next hop only if it receives content from BS.
Constraint \eqref{eqn:p4} ensures that MD $n$ can transmit content on next hop only if it receives the content on previous hop.
Constraint \eqref{eqn:p5}  ensures that each multicast group $\mathcal{K}{_{s_j}}$ should meet the minimum rate requirement $R_{min}$ to ensure quality of service. Constraint \eqref{eqn:p6} ensures that each MD can receive the content once among the total $H_{max}$ hops transmission.
Constraint \eqref{eqn:p7} is the maximum hop tolerance, which is related to the delay problem in practice.

\subsection{Complexity}
The optimization formulization that minimizes the total power consumption by optimizing multicast group division while maintaining multihop delay and data rate is an mixed integer programming (MIP) problem. MIP problem is always NP-complete due to the binary variables \cite{junger200950}. Let $\alpha$ represents the number of binary variables and  $\alpha=C^2 +H_{max}C^2$ in problem \eqref{eqn:p1}-\eqref{eqn:p7}. The worst case complexity of determining the optimal result of this MIP problem is $\mathcal{O}(2^\alpha)$.
The computational complexity of finding the optimal solution will increase exponentially as the number of MDs and transmit hops increases. Therefore, we turn to propose suboptimal methods with lower complexity in next section, which are suitable for practical communication systems.

\section{Proposed Solution - Algortithms} \label{ss2}
	
	In this section, we propose two heuristic algorithms. A core of a heuristic algorithm is to design a certain rule of choosing which MDs to connect with multicast in a multihop network. The two proposed heuristic algorithms are based on the different greedy rules and detailed in each subsection respectively.

\subsection{{Channel Gain Oriented Algorithm}}
 \label{bg}


In this algorithm, we assume that $H_{max}=C$, which is applicable for the case where the network has low-density MDs. We denote the initial set of transmitters $\mathcal{S} = \{s_1, s_2,..., s_j\}$ and the set of receivers $\mathcal{R}=\{r_1,r_2,..., r_i\}$, where $s_j$ and $r_i$ represent the $j$th transmitter and the $i$th receiver, respectively. Denote $\mathcal {K} =\{\mathcal{K}_{s_1},\mathcal{K}_{s_2},..., \mathcal {K}_{s_j}\}$ as final multicast groups and $s_j$ is the transmitter of the multicast group $\mathcal{K}_{s_j}$.
At the beginning of the algorithm, BS is the only element in $\mathcal{S}$, that is, the content delivery starts at BS. Meanwhile there are total $C$ elements as MDs in $\mathcal{R}$. The grouping of MDs is realized by following procedure: when a new link is to be established, there is only one MD with the largest $h_{{s_jr_i}}/N_0$ is chosen from the $\mathcal{R}$ as the receiver for $s_j$. The process continues until all MDs are linked. Specifically, assume that the $i$th MD in $\mathcal{R}$ and its best channel condition is the link with the $j$th MD in $\mathcal{S}$. Then let the $j$th MD be the transmitter for the $i$th MD and the $i$th MD shifts from $\mathcal{R}$ into $\mathcal{S}$ and becomes a potential transmitter for next hops. Finally the MDs with a common transmitter are divided into the same multicast group.

 Fig. \ref{fig:ii2}(a) provides an example of the grouping process. In this example, there are 3 MDs  (i.e., $s_1$, $s_2$, and $s_3$) in $\mathcal{S}$ at the beginning. A MD is selected from $\mathcal{R}$ if it has the largest channel gain with the already linked MDs. As shown in \mbox{Fig. \ref{fig:ii2}(a),} links $l_1$, ..., $l_4$ are successively established. After the whole D2D network is established as in Fig. \ref{fig:ii2}(a), the MDs with the same transmitter are grouped together for multicast as shown in Fig. \ref{fig:ii2}(b).

\begin{algorithm}[H]
\caption{Greedy Channel Gain Oriented Solution \label{eqn:pb1}}
\begin{algorithmic}[1]
\STATE Initialize $\mathcal{S}=\{BS\}$, $\mathcal{R}=\{\mathcal{C}\}$, $\mathcal{K}=\emptyset$.
\WHILE{$\mathcal{R}\neq\emptyset$}
\STATE Select $s_j\in\mathcal{S}$ and $r_i\in\mathcal{R}$ that have the maximum $h_{s_jr_i}$/$N_{0}$
\STATE Let MD $s_j$ be the transmitter for MD $r_i$
\STATE \qquad $\mathcal{K}_{s_j}\leftarrow\mathcal{K}_{s_j}\cup r_i$
\STATE Update the sets of transmitters and receivers as
\STATE \qquad $\mathcal{S}\leftarrow\mathcal{S}\cup r_{i}$
\STATE \qquad $\mathcal{R}\leftarrow\mathcal{R}\backslash r_{i}$
\ENDWHILE
\STATE Calculate power consumption of each multicast group as \eqref{ab}.
\STATE Calculate total power consumption of all groups.
\end{algorithmic}
\end{algorithm}

  \begin{figure}[H]
\centering
\subfloat[]{\includegraphics[width=4.5cm,height=4.5cm]{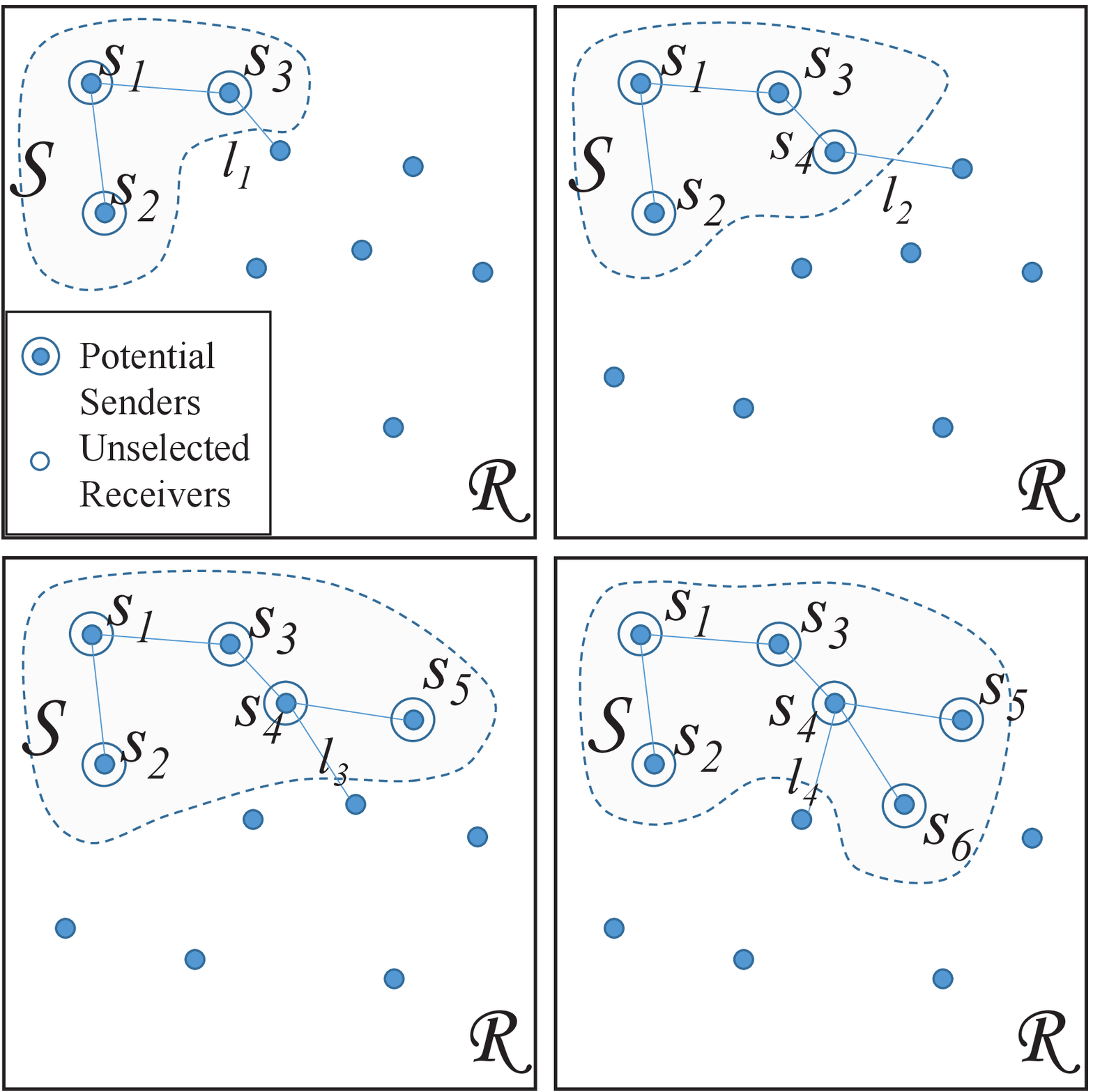}}
\subfloat[]{\includegraphics[width=4.5cm,height=4.5cm]{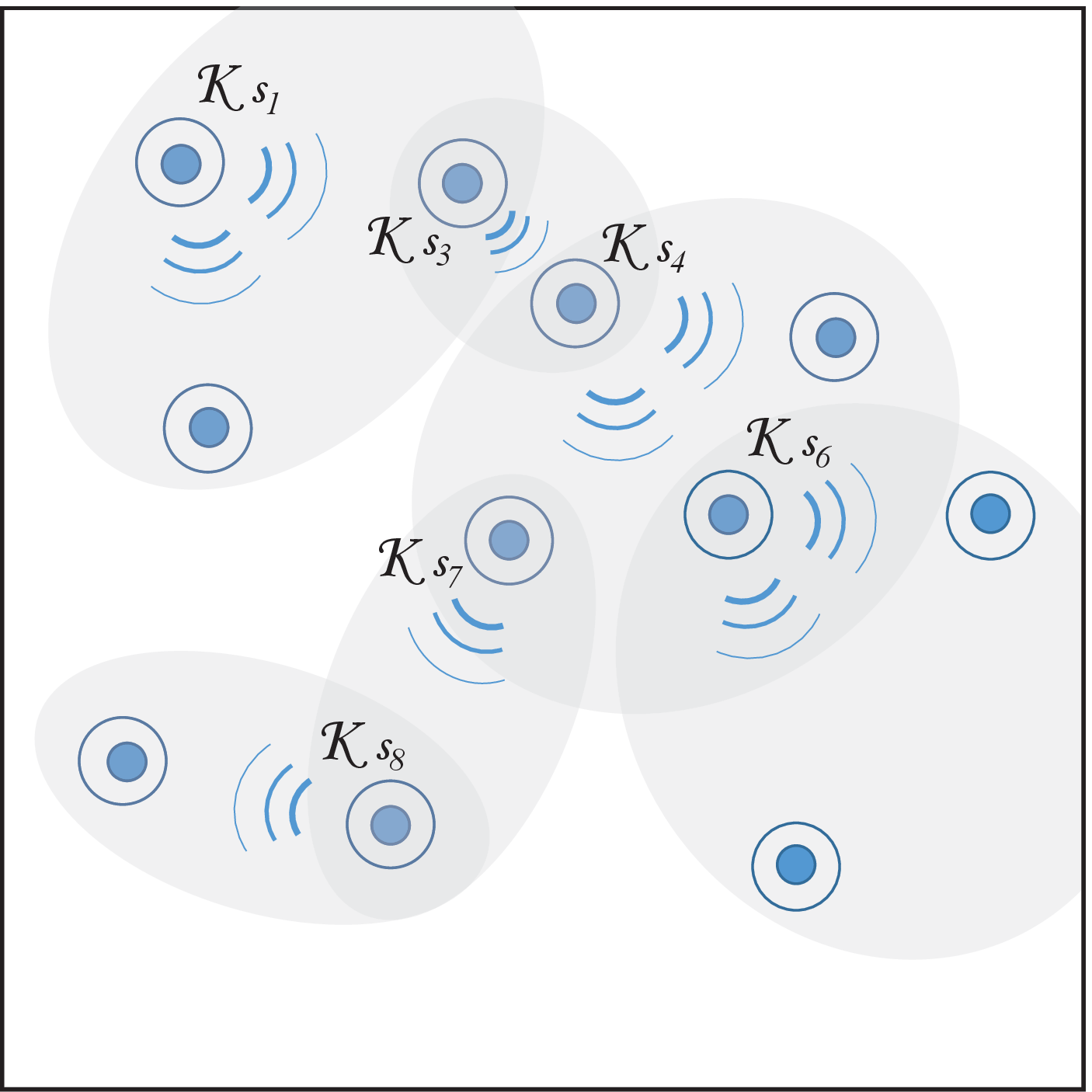}}
\caption{(a) The grouping process of Algorithm 1. $s_1, s_2,..., s_n$ are MDs in $\mathcal{S}$ acting as potential transmitters. $l_1, l_2,..., l_n$ represent the links with the maximum channel gain. (b) The multicast multihop D2D cooperative network is formed. 
}\label{fig:ii2}
 \end{figure}

 The total power consumption is the sum power consumption of all multicast groups that can be calculated by \eqref{ab}. Finally, we present the whole solution in Algorithm 1. In this algorithm, the transmitter set $\mathcal{S}$ contains at most $(C+1)$ elements and the receiver set $\mathcal{R}$ with $C$ elements. The computationally complexity of line 3 in Algorithm 1 is $\mathcal{O}(C^2)$ for finding the maximum values of channel gains. This step is repeated $C$ times and thus, the complexity of Algorithm 1 is $\mathcal{O}(C^3)$.


Note that, the channel gain oriented algorithm may trigger the delay problem if the number of MDs goes to large.
In addition, the complexity of this algorithm is $\mathcal{O}(C^3)$, which may be a little bit high for scenarios involving a large number of MDs. In the following subsection, we propose a  cluster oriented alogoithm with lower complexity.


\subsection{{Cluster Oriented Algorithm}} \label{be}
Given the needed minimum rate $R_{min}$, the proposed cluster oriented algorithm aims to decrease total power consumption by minimizing the number of multicast groups, which is in part to minimize the number of transmitters. However this problem is NP-complete even in a single hop.

\begin{theorem}The problem to determine the minimum number of transmitters on a single hop is NP-complete.
\end{theorem}

\begin{proof}
 Assume each transmitter knows itself potential receivers. For example, $\mathcal{C}$= \{\{\textbf{1},3,4\},\{\textbf{2},3,5\},\{\textbf{2},4\},\{\textbf{1},4\}\}, the MDs shown in boldface mean that they act as transmitters in subsets. Finally we choose \{\textbf{1},4\} and \{\textbf{2},3,5\} as the subsets which have the minimum number of subsets meanwhile cover all elements. This problem is similar to the set cover problem which is defined as follow: Given a set $\mathcal{A}$ and disjoint subsets $\mathcal{A}_{s_1}, \mathcal{A}_{s_2}, ... \mathcal{A}_{s_j}$ $\subseteq \mathcal{A}$, the goal is to select a minimum number of these disjoint subsets which contain all elements in $\mathcal{A}$. The set cover problem is NP-complete \cite{kozen2012design}. 
\end{proof}


The commonly used method for the set cover problem is greedy algorithm since it cannot be solved optimally in polynomial time. Here, we adopt the clustering idea into our problem, which follows the rule of selecting receivers as many as possible in a multicast group if the multicast channel gain is greater than a predefined threshold $h_{(set)}/N_0$ to satisfy the minimum rate $R_{min}$. This is because minimizing the number
\begin{figure}[H]
\centering
\subfloat[]{\includegraphics[width=4.5cm,height=4.5cm]{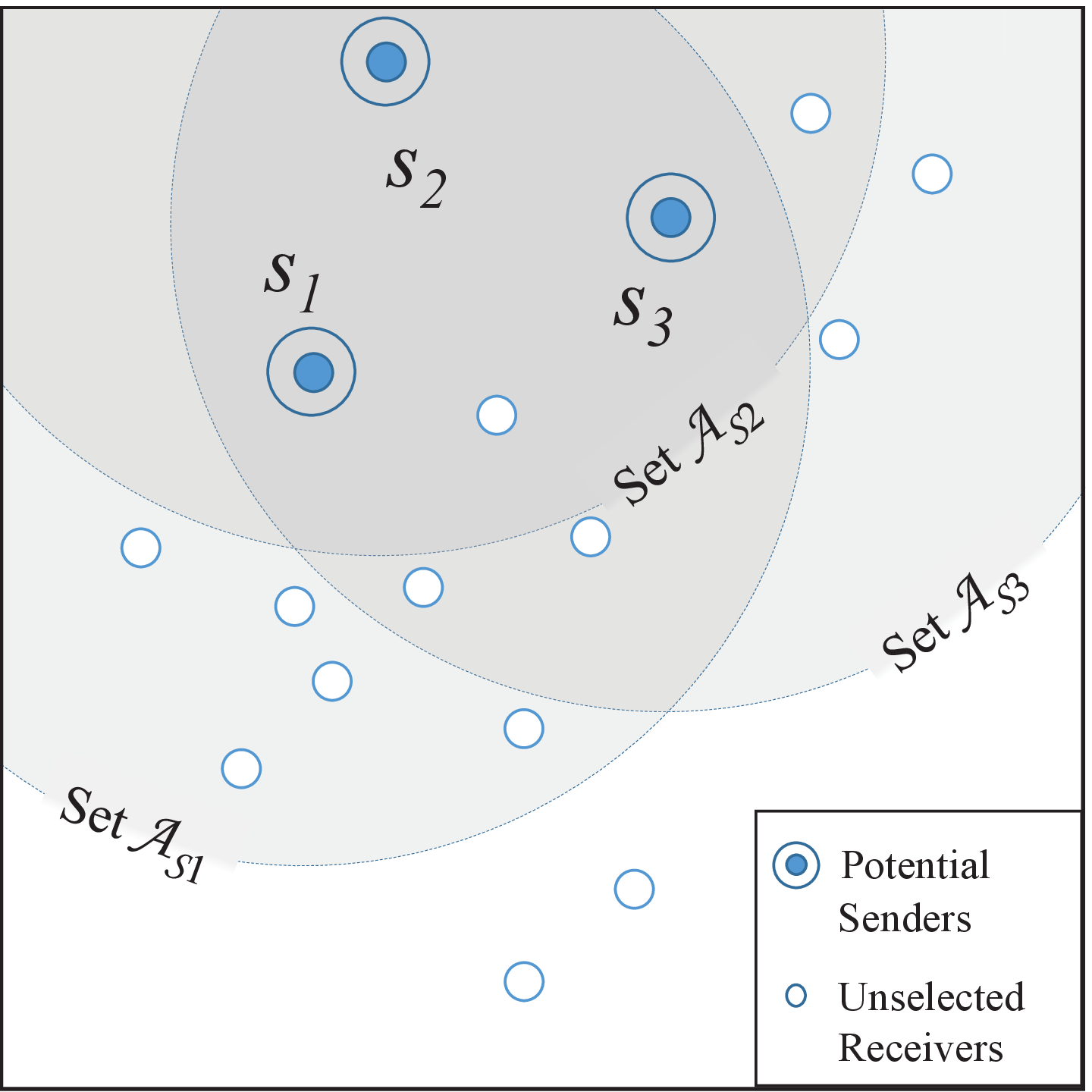}}
\subfloat[]{\includegraphics[width=4.5cm,height=4.5cm]{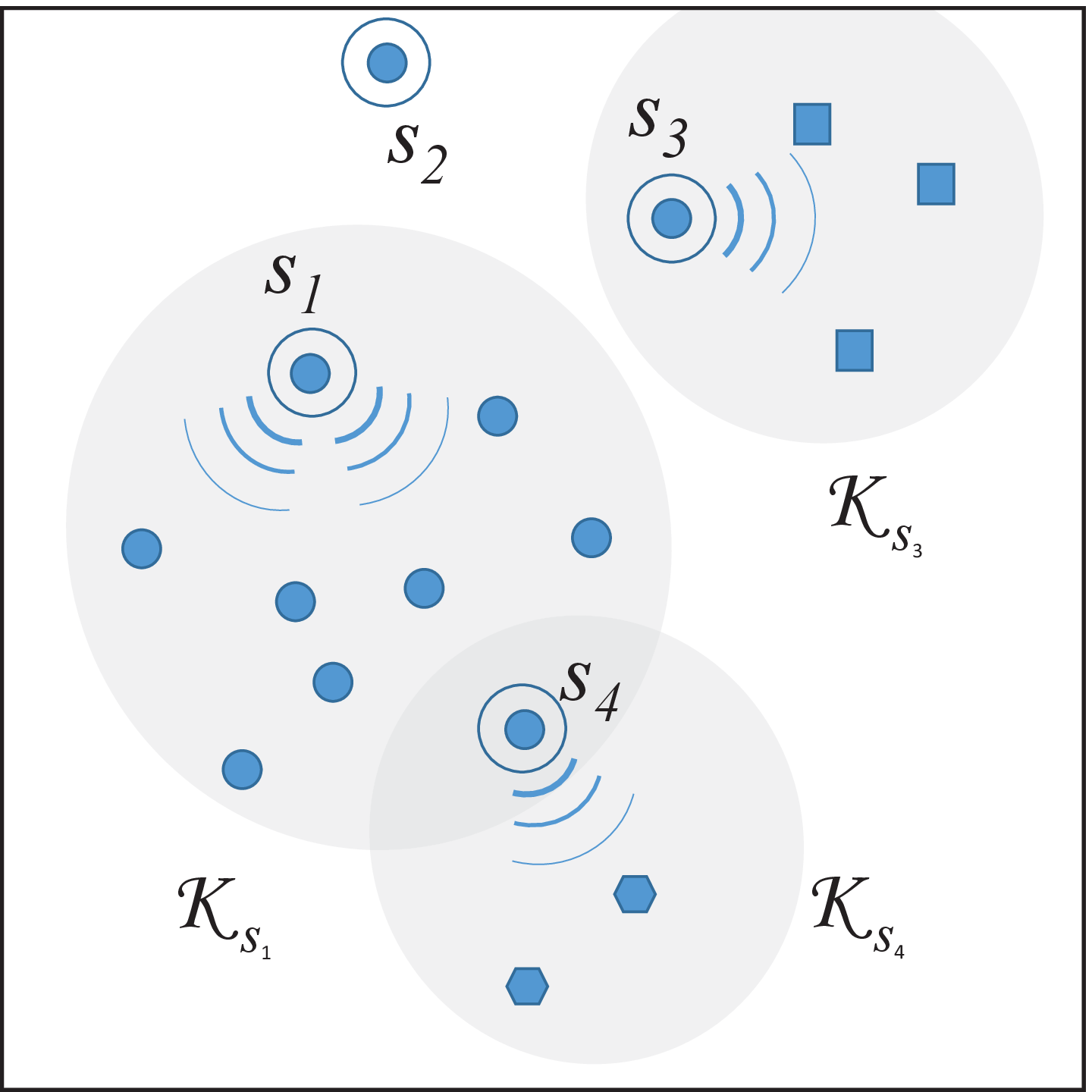}}
\caption{An example for Algorithm 2.1. (a) Every MD that has already been linked starts to find out all unlinked MDs which can meet the rate constraint and form its own set. (b) According to the set cover rule, a transmitter with the largest set of receivers is priori to form group.
}
 \label{fig:se}
 \end{figure}
\noindent of transmitters also leads to minimizing total power consumption in some sense. We define that $\mathcal{A} = \{\mathcal{A}_{s_1}, \mathcal{A}_{s_2}, ... \mathcal{A}_{s_j}\}$, where the subset $\mathcal{A}_{s_j}$ contains element $r_i$ in $\mathcal{R}$ (i.e., all uncovered MDs) if $h_{s_jr_i}/N_0 \geq h_{(set)}/N_0$. Note that it is possible that a MD may be selected by more than a transmitter (i.e., this MD is included by more than a subset of $\mathcal{A}$), so $\mathcal{A}$ is not the final grouping result.


We present the greedy set cover method in Algorithm 2.1 to solve the grouping problem on a specific hop, where $\mathcal{R}$ and $\mathcal{A}$ are inputs and return $\mathcal{K}$ as the final grouping result. At the beginning of this algorithm, $\mathcal{I}$ heritages all elements from $\mathcal{R}$, a set of all the uncovered receivers. Consequently, the intersection of $\mathcal{I}$ and $\mathcal{A}$ equals to the set of uncovered MDs which maintains the minimum rate $R_{min}$ (or equivalently the channel gain is greater than $h_{(set)}/N_0$). If there exists uncovered elements, the already covered MDs will cover the reminders on next hop. On each hop, a largest set of uncovered MDs satisfying the minimum ratet $R_{min}$ is selected to be a new multicast group $\mathcal{K}_{s_j}$ with the transmitter $s_j$. Meanwhile, $\mathcal{I}$ removes those MDs from $\mathcal{A}_{s_j}$ since they are covered. It is worth noting that after each iteration, all possible sets of \mbox{$\mathcal{I} \cap \mathcal{A}_{s_j}$} (i.e., uncovered MDs satisfying the minimum rate $R_{min}$) may be changed because some of their elements may already be removed from $\mathcal{I}$. Thus, by judging the size of $\mathcal{I} \cap \mathcal{A}_{s_j}$ instead of $\mathcal{A}_{s_j}$, the algorithm can correctly select the largest set from the reminders.

 \setcounter{algorithm}{0}
\renewcommand{\thealgorithm}{\arabic{subsection}.\arabic{algorithm}}

\begin{algorithm}[H]
\caption {Greedy Set Cover Method} 
\begin{algorithmic}[] \label{bb}
\STATE
\textbf{Step 1}. $\mathcal{I}\leftarrow\mathcal{R}$.
\STATE
\textbf{Step 2}. \textbf{while}  $\mathcal{I}\cap\mathcal{A}\ne\emptyset $
\STATE
\qquad \qquad   a: select $\mathcal{A}_{s_j}\subseteq\mathcal{A}$ that maximizes \{$\mathcal{I}\cap\mathcal{A}_{s_j}$\}
\STATE
\qquad \qquad\quad  $\mathcal{K}_{s_j}\leftarrow\mathcal{I}\cap\mathcal{A}_{s_j}$
\STATE
\qquad \qquad  b: $\mathcal{K}\leftarrow\mathcal{K}\cup\mathcal{K}_{s_j}$
\STATE
\qquad \qquad c:  $\mathcal{I}\leftarrow\mathcal{I}\backslash\mathcal{A}_{s_j}$
\STATE
\qquad\quad\, \textbf{end while}
\STATE
\textbf{Step 3}. \textbf{return} $\mathcal{K}$
\end{algorithmic}
\end{algorithm}

We take an example in Fig. \ref{fig:se} to illustrate the cluster oriented algorithm on a specific hop. Assume that transmitter set $\mathcal{S}$ contains three elements, i.e., $\mathcal{S}$ =\{$s_1$, $s_2$, $s_3$\}.  Subset $\mathcal{A}_{s_1}$, containing the most MDs which maintain the minimum rate $R_{min}$
with transmitter $s_1$, is formed and shown in Fig. \ref{fig:se}(a).

Thus, the multicast group $\mathcal{K}_{s_1}$ is formed after the first iteration. Then by removing $\mathcal {K}_{s_1}$ from the set, $\mathcal{K}_{s_3}$ is formed and thus the final grouping result is shown in Fig. \ref{fig:se}(b).

In Algorithm \ref{bc}, we present the whole algorithm to determine the grouping result in multihop.  Note that there may exist some MDs which are uncovered as they cannot meet the request of $h_{(set)}/N_0$  within $H_{max}$ hops. When such situation occurs, the algorithm reduces the channel threshold $h_{(set)}/N_0$ so that more MDs can be involved on each hop, and accordingly increases the transmit power to maintain the minimum rate constraint. The determining of $h_{(set)}/N_0$ is based on experiential simulations and we do not discuss the details in the paper.

 \begin{algorithm}[H]
\caption {{Greedy Clustering Solution}}
\begin{algorithmic}[1]\label{bc}
\STATE
Initialize $\mathcal{S}=\{BS\}$, $\mathcal{R}=\{\mathcal{C}\}$, $\mathcal{K}=\emptyset$, $h=1$.
\FOR {$h=1$ to $H_{max}$}
\STATE
{a}. Update $\mathcal{A}$ by adjusting transmit power and the predefined multicast channel threshold $h_{(set)}/N_0$ to satisfy the rate constraint.
\STATE
{b}. Run Algorithim \ref{bb} and obtain the grouping result $\mathcal{K}$ on current hop.
\STATE
{c}. Update the sets of transmitters and receivers as.\\
 \qquad \qquad\qquad   $\mathcal{S}\leftarrow\mathcal{S}\cup\mathcal{K}$\\
 \qquad \qquad\qquad   $\mathcal{R}\leftarrow\mathcal{R}\backslash\mathcal{K}$\\
{d}. Increase the number of hops $h = h + 1$
\ENDFOR
 \STATE Calculate power consumption of each multicast group as \eqref{ab}.
\STATE Calculate total power consumption of all groups.

\end{algorithmic}
\end{algorithm}

 In Algorithm 2.2, the transmitter set $\mathcal{S}$ contains at most \mbox{($C+1$)} elements and the receiver set $\mathcal{R}$ with $C$ elements. Therefore, the complexity of forming set $\mathcal{A}$ is $\mathcal{O}(C^2)$. Moreover,  line 4 needs complexity of $\mathcal{O}(C^2)$ for finding subset with maximum number of elements. Since line 3 and line 4 are repeated at most $H_{max}$ times and thus, the total complexity of the cluster oriented algorithm is $\mathcal{O}(H_{max}C^2)$.

\section{Simulation Results}
	In this section, we evaluate the performance of the proposed algorithms. We set up the stimulation parameters as follows. MDs are randomly distributed within a circle area with a radius of 500 meters, where BS locates at the center point. MDs require the same content from BS. The noise $N_0$ is considered to be -100dBm. The channel parameters are: the constant $K$ is -31.54dB, the path loss exponent is $\beta=3$, $d_0=1$m, and $\phi_{\rm dB}$  is a zero-mean Gaussian random variable which represents the effects of shadow fading \cite{goldsmith2005wireless}. The required rate \mbox{$R_{min}=10$ (bit/s/Hz).} $H_{max}$ is considered to be $C$ hops in the channel gain oriented algorithm and 10 hops in the cluster oriented algorithm. Each performance is simulated by Monte Carlo method with 10000 times.

\begin{figure}[t]
\centering
\includegraphics[scale=.6]{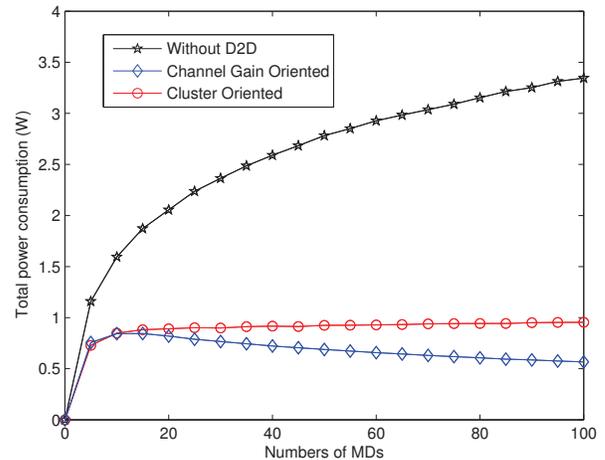}
 \caption{The total transmit power versus the number of MDs.}\label{fig:f1}
\end{figure}

\begin{figure}[t]
\centering
\includegraphics[scale=.6]{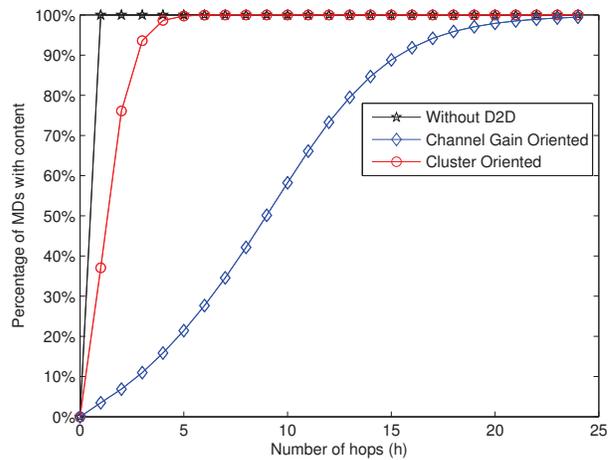}
 \caption{The percentage of MDs with the content versus the number of hops when there are 100 MDs in total.}\label{fig:2}
\end{figure}


Fig. \ref{fig:f1} shows the total transmit power consumption by the channel gain oriented and the cluster oriented algorithm. The performance of BS broadcasting is also considered as a benchmark in the figure. Obviously, the two proposed algorithms have significant performance gain compared to the traditional scheme without multihop D2D cooperation. The reason is that MDs which are close to each other can guarantee high rate and consume much less power since they have shorter distances compared to that from BS. For the channel gain oriented algorithm, we observe that the power consumption goes down as the number of MDs increases if the number of MDs is greater than 10. This mainly attributes to that the increase of density lowers D2D distances generally and thus reduces the transmit power. For the cluster oriented algorithm, the variation of power consumption is much smaller than the channel gain oriented algorithm. The channel gain oriented algorithm performs the best in total power consumption but is at the price of its uncontrolled hops. Therefore, the channel
\begin{table}[t]
\centering
\caption{\\
TOTAL POWER (W) VERSUS NUMBER OF MDS C}
\begin{tabular}{cccccccc}\hline  \label{tab:compare}
C&1&2&3&4&5&6&7\\  \hline
Algorithm 1&0.394&0.525&0.616&0.691&0.742&0.774&0.795 \\
Algorithm 2&0.394&0.528&0.624&0.695&0.734&0.766&0.799 \\
Optimal &0.394&0.509&0.573&0.638&0.654&0.709&0.739 \\ \hline
\end{tabular}
\end{table}
\noindent gain oriented algorithm is better in performance of power-saving if the content is delay-tolerant.

Fig. \ref{fig:2} shows that the channel gain oriented algorithm needs more hops than the cluster oriented algorithm. If there are 100 MDs, the channel gain oriented algorithm shows that up to 30 hops are required to make all MDs meet the rate constraint, while the cluster oriented algorithm just needs about 5 hops. This also means that remote MDs receive the content after waiting for a relatively longer time until previous MDs have obtained the content. In general, D2D networks formed by the cluster oriented algorithms have much less needed hops than that by the channel gain oriented algorithm and when the number of MDs increases, the difference continuously goes larger. As the number of hops arises, the stability of the network is undeniably affected by the multihop connections since the delay may become a problem. Such that, we can conclude that the cluster oriented algorithm is more favorable for cooperative multihop multicast in large networks if delay is an important issue.

We also compare the two proposed algorithms with the optimal situation obtained by exhaustive search  in Table \ref{tab:compare}. We can see that the performance of the two algorithms are close to each other and the difference between them is at most 0.01 W. When the number of MDs is small, the performance of the two algorithms are close to the optimal one.

\section{Conclusion}

This work addressed the energy efficiency of grouping solution in multicast multihop D2D cooperative network in a content distribution scenario. BS first multicasts content to a certain group of MD which act as relays and deliver it to other MDs in the next hops. A comprehensive optimization framework was presented for analyzing how to group MDs in order to reach a lowest total power consumption with the minimum rate constraint. This optimization problem is NP-complete. Therefore, two polynomial-time greedy algorithms were proposed to efficiently solve the problem. Simulation results showed that the total power consumption of the D2D cooperative multihop network is significantly less than that the traditional BS multicast network. In addition, the two proposed algorithms showed their individual advantages.


\appendices

\footnotesize
\bibliographystyle{IEEEtran}
\bibliography{OFDMA}

\begin{thebibliography}{10}
\providecommand{\url}[1]{#1}
\csname url@samestyle\endcsname
\providecommand{\newblock}{\relax}
\providecommand{\bibinfo}[2]{#2}
\providecommand{\BIBentrySTDinterwordspacing}{\spaceskip=0pt\relax}
\providecommand{\BIBentryALTinterwordstretchfactor}{4}
\providecommand{\BIBentryALTinterwordspacing}{\spaceskip=\fontdimen2\font plus
\BIBentryALTinterwordstretchfactor\fontdimen3\font minus
  \fontdimen4\font\relax}
\providecommand{\BIBforeignlanguage}[2]{{%
\expandafter\ifx\csname l@#1\endcsname\relax
\typeout{** WARNING: IEEEtran.bst: No hyphenation pattern has been}%
\typeout{** loaded for the language `#1'. Using the pattern for}%
\typeout{** the default language instead.}%
\else
\language=\csname l@#1\endcsname
\fi
#2}}
\providecommand{\BIBdecl}{\relax}
\BIBdecl

\bibitem{Aijaz2013A}
A.~Aijaz, H.~Aghvami, and M.~Amani, ``A survey on mobile data offloading:
  Technical and business perspectives,'' \emph{IEEE Wireless Communications},
  vol.~20, no.~2, pp. 104--112, Apr. 2013.

\bibitem{sciancalepore2016offloading}
V.~Sciancalepore, D.~Giustiniano, A.~Banchs, and A.~Hossmann-Picu, ``Offloading
  cellular traffic through opportunistic communications: Analysis and
  optimization,'' \emph{IEEE Journal on Selected Areas in Communications},
  vol.~34, no.~1, pp. 122--137, Jan. 2016.

\bibitem{kang2014mobile}
X.~Kang, Y.-K. Chia, S.~Sun, and H.~F. Chong, ``Mobile data offloading through
  a third-party {WiFi} access point: An operator's perspective,'' \emph{IEEE
  Transactions on Wireless Communications}, vol.~13, no.~10, pp. 5340--5351,
  Oct. 2014.

\bibitem{Han2012Mobile}
B.~Han, P.~Hui, V.~S.~A. Kumar, M.~V. Marathe, J.~Shao, and A.~Srinivasan,
  ``Mobile data offloading through opportunistic communications and social
  participation,'' \emph{IEEE Transactions on Mobile Computing}, vol.~11,
  no.~5, pp. 821--834, May. 2012.

\bibitem{Y2016Optimal}
Y.~Liu, ``Optimal mode selection in {D2D-Enabled} multibase station systems,''
  \emph{IEEE Communications Letters}, vol.~20, no.~3, pp. 470--473, Mar. 2016.

\bibitem{pyattaev2014network}
A.~Pyattaev, K.~Johnsson, A.~Surak, R.~Florea, S.~Andreev, and Y.~Koucheryavy,
  ``Network-assisted {D2D} communications: implementing a technology prototype
  for cellular traffic offloading,'' in \emph{Wireless Communications and
  Networking Conference (WCNC)}.\hskip 1em plus 0.5em minus 0.4em\relax IEEE,
  2014, pp. 3266--3271.

\bibitem{Al2014Optimal}
L.~Al-Kanj, H.~V. Poor, and Z.~Dawy, ``Optimal cellular offloading via
  device-to-device communication networks with fairness constraints,''
  \emph{IEEE Transactions on Wireless Communications}, vol.~13, no.~8, pp.
  4628--4643, Aug. 2014.

\bibitem{sheng2016energy}
M.~Sheng, Y.~Li, X.~Wang, J.~Li, and Y.~Shi, ``Energy efficiency and delay
  tradeoff in device-to-device communications underlaying cellular networks,''
  \emph{Selected Areas in Communications, IEEE Journal on}, vol.~34, no.~1, pp.
  92--106, Jan. 2016.

\bibitem{yun2015distributed}
J.-H. Yun and K.~G. Shin, ``Distributed coordination of co-channel femtocells
  via inter-cell signaling with arbitrary delay,'' \emph{IEEE Journal on
  Selected Areas in Communications}, vol.~33, no.~6, pp. 1127--1139, Jun. 2015.

\bibitem{Calin2010On}
D.~Calin, H.~Claussen, and H.~Uzunalioglu, ``On femto deployment architectures
  and macrocell offloading benefits in joint macro-femto deployments,''
  \emph{IEEE Communications Magazine}, vol.~48, no.~1, pp. 26 -- 32, Jan. 2010.

\bibitem{rossi2015cooperative}
C.~Rossi, C.~Casetti, C.-F. Chiasserini, and C.~Borgiattino, ``Cooperative
  energy-efficient management of federated {WiFi} networks,'' \emph{IEEE
  Transactions on Mobile Computing}, vol.~14, no.~11, pp. 2201--2215, Nov.
  2015.

\bibitem{cao2015social}
Y.~Cao, T.~Jiang, X.~Chen, and J.~Zhang, ``Social-aware video multicast based
  on device-to-device communications,'' \emph{IEEE Transactions on Mobile
  Computing}, vol.~15, no.~6, pp. 1528--1539, Jun. 2015.

\bibitem{liu2014interference}
Y.~Liu and S.~Feng, ``Interference pricing for device-to-device
  communications,'' in \emph{IEEE International Conference on Communications
  (ICC)}.\hskip 1em plus 0.5em minus 0.4em\relax IEEE, 2014, pp. 5239--5244.

\bibitem{Whitbeck2012Push}
J.~Whitbeck, Y.~Lopez, J.~Leguay, V.~Conan, and M.~D.~D. Amorim,
  ``Push-and-track: Saving infrastructure bandwidth through opportunistic
  forwarding ☆,'' \emph{Pervasive Mobile Computing}, vol.~8, no.~5, pp.
  682--697, 2012.

\bibitem{Wang2012Content}
X.~Wang, ``Content dissemination by pushing and sharing in mobile cellular
  networks: An analytical study,'' in \emph{IEEE 9th International Conference
  on Mobile Adhoc and Sensor Systems (MASS)}, 2012, pp. 353--361.

\bibitem{Ristanovic2011Energy}
N.~Ristanovic, J.~Y.~L. Boudec, A.~Chaintreau, and V.~Erramilli, ``Energy
  efficient offloading of {3G} networks,'' in \emph{Eighth IEEE International
  Conference on Mobile Ad-hoc Sensor Systems}, 2011, pp. 202--211.

\bibitem{Al2011Offloading}
L.~Al.Kanj and Z.~Dawy, ``Offloading wireless cellular networks via
  energy-constrained local ad hoc networks.'' in \emph{Global
  Telecommunications Conference (GLOBECOM), IEEE}, Dec. 2011, pp. 1--6.

\bibitem{Li2011Multiple}
Y.~Li, G.~Su, P.~Hui, D.~Jin, L.~Su, and L.~Zeng, ``Multiple mobile data
  offloading through delay tolerant networks,'' \emph{Acm International
  Workshop on Challenged Networks Las Vegas}, vol.~13, no.~7, p.~1, Jul. 2011.

\bibitem{Militano2014Wi}
L.~Militano, M.~Condoluci, G.~Araniti, and A.~Molinaro, ``{Wi-Fi} cooperation
  or {D2D}-based multicast content distribution in {LTE-A: A} comparative
  analysis,'' in \emph{ICC'14 - W13: Workshop on Cooperative and Cognitive
  Mobile Networks}, Jun. 2014, pp. 767--76.

\bibitem{Seppala2011Network}
J.~Seppala, T.~Koskela, T.~Chen, and S.~Hakola, ``Network controlled
  {Device-to-Device (D2D)} and cluster multicast concept for {LTE} and {LTE-A}
  networks,'' in \emph{Wireless Communications and Networking Conference
  (WCNC)}, 2011, pp. 986--991.

\bibitem{junger200950}
M.~J{\"u}nger, T.~M. Liebling, D.~Naddef, G.~L. Nemhauser, W.~R. Pulleyblank,
  G.~Reinelt, G.~Rinaldi, and L.~A. Wolsey, \emph{50 Years of Integer
  Programming 1958-2008: From the Early Years to the State-of-the-art}.\hskip
  1em plus 0.5em minus 0.4em\relax Springer Science \& Business Media, 2009.

\bibitem{kozen2012design}
D.~C. Kozen, \emph{The design and analysis of algorithms}.\hskip 1em plus 0.5em
  minus 0.4em\relax Springer Science Business Media, 2012.

\bibitem{goldsmith2005wireless}
A.~Goldsmith, \emph{Wireless communications}.\hskip 1em plus 0.5em minus
  0.4em\relax Cambridge university press, 2005.

\end{thebibliography}

\end{document}